\documentclass[draftclsnofoot, onecolumn, 11pt, two-side]{IEEEtran}
\usepackage{times,amsmath,color,amssymb,graphicx,epsfig,psfrag,enumerate,amsthm,multicol,subfigure}
\usepackage{cite}
\usepackage[square,comma,numbers]{natbib}
\usepackage[top=0.8in,bottom=0.8in,left=0.8in,right=0.8in]{geometry}
\theoremstyle{plain}
\newtheorem{Theorem}{Theorem}
\begin{document}

\title{Compress-and-Forward Performance in Low-SNR Relay Channels}

\author{Lili~Zhang, %~\IEEEmembership{Student Member,~IEEE,}
        Jinhua~Jiang, %~\IEEEmembership{Member,~IEEE,}
        %Andrea J. Goldsmith,
        and~Shuguang~Cui%,~\IEEEmembership{Member,~IEEE}% <-this % stops a space
\thanks{L. Zhang and S. Cui are with the Department
of Electrical and Computer Engineering, Texas A$\&$M University,
College Station, TX 77843 USA
(e-mail: lily.zhang@tamu.edu; cui@ece.tamu.edu).}
% <-this % stops a space
\thanks{J. Jiang is with the Department of Electrical Engineering, Stanford University, Stanford, CA 94305 USA (email: jhjiang@stanford.edu).}
% <-this % stops a space
}
%\author{\IEEEauthorblockN{Lili Zhang}
%\IEEEauthorblockA{Department of ECE\\
%Texas A\&M University\\
%College Station, TX 77843, USA\\
%Email: lily.zhang@tamu.edu } \\\and \IEEEauthorblockN{Jinhua
%Jiang}
%\IEEEauthorblockA{Department of EE\\ Stanford University\\
%Stanford, CA 94305, USA\\
%Email: jhjiang@stanford.edu }\\ \and \IEEEauthorblockN{Shuguang
%Cui}
%\IEEEauthorblockA{Department of ECE\\ Texas A\&M University\\
%College Station, TX 77843, USA\\
%Email: cui@ece.tamu.edu }
%}

\maketitle

\begin{abstract}
In this paper, we study the Gaussian relay channels in the low signal-to-noise ratio (SNR) regime with the time-sharing compress-and-forward (CF) scheme, where at each time slot all the nodes keep silent at the first fraction of time and then transmit with CF at a higher peak power in the second fraction. Such a silent vs. active two-phase relay scheme is preferable in the low-SNR regime. With this setup, the upper and lower bounds on the minimum energy per bit required over the relay channel are established under both full-duplex and half-duplex relaying modes. In particular, the lower bound is derived by applying the max-flow min-cut capacity theorem; the upper bound is established with the aforementioned time-sharing CF scheme, and is further minimized by letting the active phase fraction decrease to zero at the same rate as the SNR value. Numerical results are presented to validate the theoretical results.
\end{abstract}

\section{Introduction} \label{sec_1}

The relay channel problem was first studied by van der Meulen \cite{van_der}, where the source transmits a message to the destination with the help of a relay node. In \cite{cover}, the authors formalized the problem and proposed several effective full-duplex relay strategies, which are the so-called decode-and-forward (DF) and compress-and-forward (CF) schemes. Subsequently, the corresponding schemes in the time division half-duplex (TDD) mode were studied in \cite{anders}. The upper bound on the capacity and various achievable rates were established in \cite{cover} and \cite{anders}, while the general capacity problem is still open.

One of the important applications of relaying is in low-power communication networks (such as wireless sensor networks), where the transmitting SNRs are usually low such that relays are needed to enhance the performance at the destination nodes. In \cite{verdu}, the author presented two important performance criteria in the low-SNR regime: One is the minimum energy per bit required for reliable communication; and the other is the slope of the spectral efficiency. Such criteria have been adopted to study the relay channel in the low-SNR regime~\cite{gamal,xiaodong}. In particular, the time-sharing CF scheme and the traditional DF scheme were investigated in \cite{gamal} for both the full-duplex and frequency division half-duplex (FDD) modes, where the authors derived the DF-based upper bound on the minimum energy per bit, and the lower bound using the max-flow min-cut capacity theorem. In \cite{xiaodong}, the fading relay channel was studied, where the upper and lower bounds on the minimum energy per bit were derived using the same method as in \cite{gamal}.

In the low-SNR regime, the CF scheme is much more appealing than the DF scheme since the relay cannot decode much information such that the DF strategy at the relay would cause the bottleneck effect. In this paper, we thus focus on the time-sharing CF scheme in the low-SNR regime under both the full-duplex and half-duplex TDD modes. We show that the time-sharing CF scheme improves the achievable rate under the TDD mode when the SNR is low. Furthermore, the lower bound on the minimum energy per bit is derived by applying the max-flow min-cut theorem for the channel capacity, and the upper bound is derived by deploying the time-sharing CF scheme. The aforementioned upper bound can be further tightened by letting the active phase duration decrease to zero at the same rate as the SNR value. %We also show the performance improvement due to time-sharing by numerical results.

The rest of the paper is organized as follows. The channel model is discussed in Section \ref{sec_2}. In Section \ref{sec_3} we present the achievable rates of the time-sharing CF scheme under both the full-duplex and half-duplex TDD modes, and subsequently we derive the bounds on the minimum energy per bit for reliable communication. The numerical results are presented in Section \ref{sec_4}. Finally, the paper is concluded in Section \ref{sec_5}.

\section{System Model}\label{sec_2}
Consider the Gaussian relay channel as shown in Fig. \ref{fig_relay_channel_model}, where the channel amplitude gains are assumed fixed and denoted as $h_{21}$, $h_{31}$, and $h_{32}$, respectively. Next, we introduce the signal models for both the traditional relaying schemes (full-duplex and half-duplex TDD), and the time-sharing relaying scheme. Note that in this paper we only consider real signals since the extension to complex cases is straightforward.
%\begin{figure}[!t]
%\centering
%% for double column submission
%\includegraphics[width=3in]{./figure/relay_channel_model.eps}
%\caption{The Relay Channel Model.} \label{fig_relay_channel_model}
%\end{figure}

\subsection{Traditional Relaying}

\subsubsection{Full-Duplex Mode}
Suppose the relay can transmit and receive signals at the same time. The channel inputs at the source and relay are represented as $X_1$ and $X_2$ with average power values $P_1$ and $P_2$, respectively. The noises at the relay and destination are denoted by $Z_1$ and $Z$, which are independent additive white Gaussian noises (AWGNs) with distributions $\mathcal{N}(0,N_1)$ and $\mathcal{N}(0,N)$, respectively.

Denoting the received signals at the relay and destination by $Y_1$ and $Y$, respectively, we can fully describe the channel as:
\begin{eqnarray*}
Y_1 &=& h_{21}X_{1}+Z_1,\\
Y &=& h_{31}X_{1} + h_{32}X + Z.
\end{eqnarray*}

\subsubsection{Half-Duplex Mode}
We also consider the half-duplex TDD mode, where at each transmission slot the relay listens to the source during the first $\lambda$ fraction of time, and transmits to the destination in the remaining $1-\lambda$ fraction.

In the first $\lambda$ fraction, the transmitted signal at the source is denoted by $X_{1}^{(1)}$ with an average power constraint $P_{1}^{(1)}$. The received signals at the relay and destination are given by:
\begin{eqnarray*}
Y_{1} &=& h_{21}X_{1}^{(1)} + Z_{1},\\
Y^{(1)} &=& h_{31}X_{1}^{(1)} +
Z^{(1)},
\end{eqnarray*}
respectively, where $Z_{1} \sim \mathcal{N}(0,N_{1})$ and $Z^{(1)} \sim \mathcal{N}(0,N)$ are independent AWGNs.

In the remaining $1-\lambda$ fraction, the transmitted signals at the source and relay are denoted as $X_{1}^{(2)}$ and $X_{2}$,
with average power constraints $P_{1}^{(2)}$ and $P_{2}$, respectively. Accordingly, the received signal at the destination is given by:
\begin{eqnarray*}
%Y_{1}(i) &=& h_{21}X_{1}^{(1)}(i) + Z_{1}(i),\\
Y^{(2)} &=& h_{31}X_{1}^{(2)} +
h_{32}X_{2} + Z^{(2)},
\end{eqnarray*}
where $Z^{(2)} \sim \mathcal{N}(0,N)$ is the AWGN noise at the destination.

\subsection{Relaying with Time-Sharing}
It is known that for low-SNR point-to-point Gaussian channels, low-duty-cycle on-off keying is as good as Gaussian inputs~\cite{verdu}. Along a similar line, for the low-SNR relay channels, the authors in~\cite{gamal} proposed a time-sharing CF scheme based on the argument that the CF achievable rate is not a concave function over the power allocation factor, where the underlying idea is that the relay can ``hear" the source more clearly if we first deploy a silent phase in each slot and increase the peak power in the following active phase.

With the above discussions, in this paper we consider the time-sharing relaying scheme, where in each time slot all the nodes keep silent in the first $1-\alpha$ fraction (called the silent phase), and deploy the traditional relaying scheme in the remaining $\alpha$ fraction (called the active phase), which could be either full-duplex or half-duplex as introduced in the previous subsection. As such, the signal model in the active phase follows the previously introduced model except that the power at the source and relay should be scaled up accordingly to meet the average power constraint. For example, if full-duplex relaying is deployed, the power values at the source and relay become $P_1/\alpha$ and $P_2/\alpha$, respectively.

\section{Achievable Rates and Minimum Energy Per Bit}\label{sec_3}
In this section, the performance of the time-sharing CF scheme is investigated under both the full-duplex and half-duplex TDD modes. Specifically, we derive the achievable rates and the upper/lower bounds on the minimum energy per bit in the low-SNR regime. For convenience, we define the following function:
\[ \Gamma \left(x\right) =
\frac{1}{2} \log_{2} \left(1 + x\right),
\]
and denote the \emph{normalized channel gains} of the corresponding links
as:
\begin{equation}
\nonumber \gamma_{21} = \frac{h_{21}^{2}}{N_{1}},\quad \gamma_{32} =
\frac{h_{32}^{2}}{N}, \quad \gamma_{31} = \frac{h_{31}^{2}}{N}.
\end{equation}

\subsection{Capacity Upper Bound and Achievable Rates}
\subsubsection{Full-Duplex Mode}
We first give the upper bound on the capacity by applying the max-flow min-cut theorem \cite{cover}:
\begin{equation}
C^{+}_{\textrm{full}} = \max_{0\leq \rho_x \leq
1}\min\left\{C^{+}_{\textrm{1-full}}(\rho_x),
C^{+}_{\textrm{2-full}}(\rho_x)\right\}, \label{eq_full_upperbound}
\end{equation}
where $\rho_x$ is the correlation coefficient between the source and relay inputs, and
\begin{eqnarray}
\nonumber C^{+}_{\textrm{1-full}}(\rho_x) &=&
\Gamma\left(\gamma_{31}P_1 +
\gamma_{32}P_2\right), \\
\nonumber C^{+}_{\textrm{2-full}}(\rho_x) &=&
\Gamma\left(P_1(1-\rho_x^2)(\gamma_{21}+\gamma_{31}) \right).
\end{eqnarray}

Now we first consider the traditional CF strategy, where the relay compresses the received signal with Wyner-Ziv coding~\cite{wyner}, and then forwards the binning index to the destination. The following achievable rate is established with the CF scheme \cite{cover}:
\[
R_{\textrm{cf}} = \sup_{p(\cdot)\in\mathcal{P}^{*}} I(X_1;Y,\hat{Y}_1|X_2),
\]
subject to the constraint
\[
I(X_2;Y) \geq I(Y_1;\hat{Y}_1|X_2,Y).
\]
Correspondingly, the achievable rate with CF over the Gaussian relay channel can be written as:
\begin{equation}
R_{\textrm{cf}} = \Gamma \left(\gamma_{31}P_1 + \frac{\gamma_{32}\gamma_{21}P_1P_2}{1+\gamma_{21}P_1+\gamma_{31}P_1+\gamma_{32}P_2}\right).
\end{equation}
As we discussed before, to improve the performance in the low-SNR regime, the authors in~\cite{gamal} proposed the time-sharing CF scheme based on the argument that the achievable rate is not a concave function of the power. Specifically, the relay channel is only utilized for $\alpha$ portion of each time slot.  In this $\alpha$ portion, the power values of the source and relay are $P_1/\alpha$ and $P_2/\alpha$, respectively. Accordingly, the achievable rate for the time-sharing CF scheme is given as~\cite{gamal}:
\begin{equation}
R_{\textrm{ts}} = \max_{0<\alpha\leq 1} \alpha\Gamma \left(\frac{\gamma_{31}P_1}{\alpha} + \frac{\gamma_{32}\gamma_{21}P_1P_2}{\alpha^2+\alpha(\gamma_{21}P_1+\gamma_{31}P_1+\gamma_{32}P_2)}\right).
\end{equation}
It is also shown in \cite{gamal} that the above achievable rate is higher than the one with the traditional CF scheme \cite{cover} in the low-SNR regime.

\subsubsection{Half-Duplex Mode}
In practice, full-duplex relaying is technically difficult to implement. Hence, we now consider the more practical TDD mode. We first give the max-flow min-cut bound on the capacity \cite{anders}:
\begin{equation}
C^{+}_{\textrm{half}} = \max_{0 \leq \rho_{x} \leq 1} \min
\{C_{\textrm{1-half}}^{+}\left(\rho_{x}\right),
C_{\textrm{2-half}}^{+}\left(\rho_{x}\right)\},\label{eq_half_upperbound}
\end{equation}
where $\rho_x$ is the correlation coefficient between $X_1^{(2)}$ and $X_2$, and
\begin{eqnarray*}
C_{\textrm{1-half}}^{+}\left(\rho_{x}\right)&=& \lambda \Gamma
\left(\gamma_{31}P_1^{(1)}\right) + \\  &&\left(1-\lambda\right) \Gamma
\bigg(\gamma_{31}P_1^{(2)} + %\\&&
\gamma_{32}P_2 + 2\rho_{x}\sqrt{\gamma_{31}P_1^{(2)}\gamma_{32}P_2}\bigg); \nonumber \\
C_{\textrm{2-half}}^{+}\left(\rho_{x}\right)&=& \lambda \Gamma
\left(\left(\gamma_{21}+\gamma_{31}\right)P_1^{(1)}\right)+\nonumber %\\&&
\\&&\ \left(1-\lambda\right) \Gamma
\left(\left(1-\rho_x^2\right)\gamma_{31}P_1^{(2)}\right).
\end{eqnarray*}

Furthermore, the traditional CF achievable rate without time-sharing under the half-duplex TDD mode is given by \cite{anders}:
\begin{equation}
R_{\textrm{cf\_h}} = \lambda \Gamma \left(\gamma_{31}P_1^{(1)} + \frac{\gamma_{21}P_1^{(1)}}{1+N_{\textrm{w}}/N_1}\right) + (1-\lambda)\Gamma \left(\gamma_{31}P_1^{(2)}\right),
\end{equation}
where $N_{\textrm{w}}$ is the power of the quantization noise at the relay, with
\begin{equation}
\nonumber N_{\textrm{w}} = N_1\frac{1+\gamma_{21}P_1^{(1)}+\gamma_{31}P_1^{(1)}}{\left(1+\gamma_{31}P_1^{(1)}\right)\left(\left(1+\frac{\gamma_{32}P_2}{1+\gamma_{31}P_1^{(2)}}\right)^{\frac{1-\lambda}{\lambda}}-1\right)}.
\end{equation}

In the TDD mode, the above achievable rate function is not a concave function of the power either. Thus, by a similar argument as that in \cite{gamal}, we can apply time-sharing in the low-SNR regime to improve the achievable rate. In particular, the source and relay keep silent in the first $1-\alpha$ portion of each time slot; then the source sends the message to the relay and destination in the following $\alpha\lambda$ portion with power $P_1^{(1)}/\alpha$. The relay compresses the received signal and sends it to the destination over the last $\alpha(1-\lambda)$ fraction of time with power $P_2/\alpha$. The source also utilizes the last $\alpha(1-\lambda)$ portion to send signal to the destination with power $P_1^{(2)}/\alpha$. The achievable rate of the above time-sharing CF scheme under the half-duplex TDD mode can be computed as:
{%\setlength{\arraycolsep}{2pt}
\begin{equation}
\scriptsize\label{eq_half_cf_ts_rate}
\nonumber R_{\textrm{ts\_h}} = \max_{0<\alpha\leq 1}\quad \alpha\lambda \Gamma \left(\frac{\gamma_{31}P_1^{(1)}}{\alpha} + \frac{\gamma_{21}P_1^{(1)}}{\alpha+\alpha \frac{N_{\textrm{ts\_w}}}{N_1}}\right)   +\  \alpha(1-\lambda)\Gamma \left(\frac{\gamma_{31}P_1^{(2)}}{\alpha}\right),
\end{equation}}
where $N_{\textrm{ts\_w}}$ is the power of the quantization noise at the relay, with
\begin{equation}
\nonumber N_{\textrm{ts\_w}} = N_1\frac{\alpha+\gamma_{21}P_1^{(1)}+\gamma_{31}P_1^{(1)}}{\left(\alpha+\gamma_{31}P_1^{(1)}\right)\left(\left(1+\frac{\gamma_{32}P_2}{\alpha+\gamma_{31}P_1^{(2)}}\right)^{\frac{1-\lambda}{\lambda}}-1\right)}.
\end{equation}

The achievable rate improvement in the low-SNR regime under both the full-duplex and half-duplex modes will be further illustrated in Section \ref{sec_4}, by numerical comparisons against the schemes without time-sharing.

\subsection{Minimum Energy Per Bit}
In \cite{verdu}, the author presented two important performance criteria in the low-SNR regime: One is the minimum energy per bit required for reliable communication, and the other is the slope of the spectral efficiency. In this paper, we focus on the first criterion. Denoting the channel capacity as $C(P)$ with $P$ the total power consumed at the source and the relay, the minimum energy per bit is defined as~\cite{verdu}:
\begin{equation}
\left(\frac{E_b}{N_0}\right)_{\min}=\lim_{P \to 0}\frac{P}{C(P)}=\frac{\log_e 2}{C'(0)},
\end{equation}
where $C'(0)$ is the first-order derivative at $P = 0$. Note that here the energy per bit is referring to the total energy consumed at both the source and the relay for each bit transmitted over the relay channel.

%If two channels have the same minimum energy per bit, then the following spectral efficiency becomes the key performance measure in the low-SNR regime:%, which is related to the second order derivative of $C(\textrm{SNR})$,
%\begin{equation}
%S_0 = \frac{ 2 [C'(0)]^2}{-C''(0)},
%\end{equation}
%where $C''(0)$ is the second order derivative at $0$ of $C(\textrm{SNR})$.

In the rest of this section, we derive the upper and lower bounds on the minimum energy per bit under both the full-duplex and half-duplex modes for the Gaussian relay channel under consideration.

\subsubsection{Full-Duplex Mode}
We skip the derivation of the lower bound on the minimum energy per bit under the full-duplex mode, which is obtained by applying the max-flow min-cut upper bound on the channel capacity. Interested readers can refer to \cite{gamal} and \cite{xiaodong} for related results. Now we derive the upper bound on the minimum energy per bit.

Assume that the sum power across the source and the relay is subject to a total power constraint $P$. Denote the power at the source as $\beta P$ and the power at the relay as $(1-\beta)P$ for $0<\beta\leq 1$, i.e., $\{P_1,P_2\} = \{\beta P, (1-\beta) P\}$. The achievable rate of the time-sharing CF scheme can be calculated as:
\begin{equation}
R_{\textrm{ts}} = \max_{0<\alpha\leq 1} \alpha\Gamma\left(\frac{\gamma_{31}\beta P}{\alpha}+\frac{\gamma_{32}\gamma_{21}\beta(1-\beta)P^2}{\alpha^2+\alpha GP}\right),
\end{equation}
where $G:=\gamma_{21}\beta+\gamma_{31}\beta+\gamma_{32}(1-\beta)$. For a given active fraction $\alpha$, the upper bound on the minimum energy per bit can be evaluated as:
\begin{eqnarray}
\nonumber\left(\frac{E_b}{N_0}\right)_{\min}^{+} &=& \lim_{P \to 0} \frac{P}{R_{\textrm{ts}}(P)}\\
&=& \lim_{P \to 0} \frac{P}{\alpha}\frac{1}{\Gamma\left(\gamma_{31}\beta P/\alpha+\frac{\gamma_{32}\gamma_{21}\beta(1-\beta)P^2/\alpha^2}{1+GP/\alpha}\right)}.\label{eq_full_ebn0_upper}
\end{eqnarray}
%first order derivative of the achievable rate over $P$ is:
%\begin{equation}\label{eq_full_deri}
%R'_{\textrm{ts}}(P) = \frac{1}{2}\frac{\gamma_{31}\beta+\frac{\gamma_{32}\gamma_{21}\beta(1-\beta)(2\alpha^2P+\alpha GP^2)}{\alpha(\alpha+GP)^2}}{1+\gamma_{31}\beta \frac{P}{\alpha} + \frac{\gamma_{32}\gamma_{21}\beta(1-\beta)P^2}{\alpha(\alpha+GP)}}.
%\end{equation}
To determine the limit, now we discuss several interesting cases based on the behavior of the parameter $\alpha$. When the power $P$ decreases to $0$, while $\alpha$ keeps constant or does not scale down as fast as $P$, the upper bound on the minimum energy per bit in \eqref{eq_full_ebn0_upper} is equal to $2\log_e 2/(\gamma_{31}\beta)$, which is minimized at $\beta=1$. In this case, the power is all utilized by the source, and the upper bound on the minimum energy per bit is the same as that of the point-to-point communication channel. The above result also implies that such an upper bound corresponding to the traditional CF scheme with a fixed $\alpha=1$ is equal to $2\log_e 2/\gamma_{31}$. Furthermore, if the active fraction $\alpha$ goes to $0$ faster than $P$, $\left(E_b/N_0\right)_{\min}^{+}$ in \eqref{eq_full_ebn0_upper} goes to infinity, which leads to a non-meaningful upper bound on the minimum energy per bit. However, if we keep the ratio $P/\alpha$ constant (denoted by $A$) when $P$ and $\alpha$ decrease to zero, we could achieve a tighter upper bound (against the case of $\alpha=1$) on the minimum energy per bit.
%Denote $A:=\lim_{P\to 0}P/\alpha$, the first order derivative of $R_{\textrm{ts}}$ at $P=0$ can be expressed as:
%\begin{eqnarray}
%\small
%R'_{\textrm{ts}}(0) &=& \lim_{P\to 0, \alpha\to 0} R'_{\textrm{ts}}(P)\\
%\nonumber &=& \frac{1}{2}\frac{\gamma_{31}\beta(1+AG)^2+\gamma_{32}\gamma_{21}\beta(1-\beta)(2A+A^2G)}{(1+\gamma_{31}\beta A)(1+AG)^2+\gamma_{32}\gamma_{21}\beta(1-\beta)A^2(1+AG)}.
%\end{eqnarray}
Therefore, we have the following theorem describing the asymptotic upper bound on the minimum energy per bit using the time-sharing CF scheme.

\begin{Theorem}
Under the full-duplex mode, when both $P$ and $\alpha$ tend to zero with $\frac{P}{\alpha}=A$, the minimum energy per bit on the relay channel has the following asymptotic upper bound:
\begin{equation}\label{eq_full_ebn0_upper_theorem}
\left(\frac{E_b}{N_0}\right)_{\min}^{+} = \min_{0<\beta\leq 1, A\geq0} \frac{A}{\Gamma\left(\gamma_{31}\beta A+ \frac{\gamma_{32}\gamma_{21}\beta(1-\beta)A^2}{1+GA}\right)}.
\end{equation}
\end{Theorem}

It can be shown that under the full-duplex mode, the above upper bound converges to the lower bound on the minimum energy per bit when $\gamma_{32}\to\infty$~\cite{gamal} since both the upper and lower bounds go to $2\log_e 2/(\gamma_{31}+\gamma_{21})$.
The upper bound improvement of the minimum energy per bit due to time-sharing will be further illustrated in Section \ref{sec_4} by numerical comparisons in the low-SNR regime.

\subsubsection{Half-Duplex Mode}
Here we investigate the lower and upper bounds on the minimum energy per bit under the half-duplex mode. The sum of the source power and the relay power is fixed to $P$ as in the full-duplex mode. Denote the power at the source as $\beta P$ and the power at the relay as $(1-\beta)P$, respectively. Furthermore, with CF we assume that the source uses the same power during the first and the second transmitting phases, i.e., $P_1^{(1)} = P_1^{(2)}=\beta P$. Consequently, we have the power allocation as $\{P_1^{(1)},P_1^{(2)},P_2\}=\{\beta P,\beta P, (1-\beta)P\}$.

First we derive the lower bound on the minimum energy per bit, which is given by the following theorem.
%The detailed proof is given in Appendix \ref{appendix_1}.

\begin{Theorem}
Under the half-duplex TDD mode, the lower bound on the minimum energy per bit can be obtained by one-dimensional searching over $\lambda \in (\gamma_{31}/(\gamma_{31}+\gamma_{21}),1]$ instead of three-dimensional searching over $\lambda$, $\beta$, and $\rho_x$. When the relay-destination link gain $\gamma_{32}$ goes to infinity, the lower bound converges to $2\log_e 2/(\gamma_{31}+\gamma_{21})$.
%\begin{equation}\label{eq_half_lowerbound_2}
%\nonumber \left(\frac{E_b}{N_0}\right)_{\min}^{-} = \left\{\begin{array}{ll}
%\frac{2(\gamma_{32}-4\gamma_{21})\log_e 2 }{\gamma_{31}\gamma_{32}-2\gamma_{21}\gamma_{31}+\gamma_{21}\gamma_{32}-2\gamma_{21}\sqrt{\gamma_{31}^2+\gamma_{31}\gamma_{32}+\gamma_{21}\gamma_{32}}} & \textrm{if}\ \gamma_{32}\ \textrm{is finite,}\\
%\frac{2\log_e 2}{\gamma_{31}+\gamma_{21}} & \textrm{otherwise.}
%\end{array}\right.
%\end{equation}
\end{Theorem}
\begin{proof}
See Appendix \ref{appendix_1}.
\end{proof}

Now we consider the upper bound on the minimum energy per bit, which is achieved by the time-sharing CF scheme. %In order to simplify the derivations, we fix the TDD parameter $\lambda$ to $0.5$.

\begin{Theorem}
Under the half-duplex mode, the minimum energy per bit on the relay channel is upper-bounded by:
\begin{equation}\label{eq_half_ebn0_upper}\small
\left(\frac{E_b}{N_0}\right)_{\min}^{+} = \min_{0<\lambda\leq 1,0<\beta\leq 1,A\geq 0}\frac{A}{\lambda\Gamma\left(\gamma_{31}\beta A + \frac{\gamma_{21}\beta A}{1+\frac{1+\gamma_{21}\beta A+\gamma_{31}\beta A}{(1+\gamma_{31}\beta A)\left(\left(1+\frac{\gamma_{32}(1-\beta)A}{1+\gamma_{31}\beta A}\right)^{(1-\lambda)/\lambda}-1\right)}}\right)+(1-\lambda)\Gamma\left(\gamma_{31}\beta A\right)},\ \textrm{with}\ A=\frac{P}{\alpha}.
\end{equation}
%\begin{equation}\label{eq_half_ebn0_upper}
%\left(\frac{E_b}{N_0}\right)_{\min}^{+} = \frac{2\log_e 2}{\max_{0<\beta\leq 1,A} \left(R'_{\textrm{ts}}(0)+\frac{1}{2}\frac{\gamma_{31}\beta}{1+\gamma_{31}\beta A}\right)},
%\end{equation}
Such an upper bound is tight (i.e., achieves the lower bound on the minimum energy per bit) asymptotically when $\gamma_{32}$ goes to infinity.
\end{Theorem}

\begin{proof}
As shown in \eqref{eq_half_cf_ts_rate}, the achievable rate of the time-sharing CF scheme under the half-duplex mode is computed as:
\begin{equation}
\small
R_{\textrm{ts\_h}} = \max_{0<\alpha\leq 1} \alpha\lambda \Gamma \left(\frac{\gamma_{31}\beta P}{\alpha} + \frac{\gamma_{21}\beta P}{\alpha+\alpha N_{\textrm{ts\_w}}/N_1}\right)   +\  \alpha(1-\lambda)\Gamma \left(\frac{\gamma_{31}\beta P}{\alpha}\right),\label{eq_half_cf_ts}
\end{equation}
where the quantization noise power $N_{\textrm{ts\_w}}$ is
%\begin{equation}
%N_{\textrm{ts\_w}} = N_1\frac{\alpha+\gamma_{21}\beta P+\gamma_{31}\beta %P}{\gamma_{32}(1-\beta)P}.\label{eq_half_noise}
%\end{equation}
\begin{equation}
\nonumber N_{\textrm{ts\_w}} = N_1\frac{\alpha+\gamma_{21}P_1^{(1)}+\gamma_{31}P_1^{(1)}}{\left(\alpha+\gamma_{31}P_1^{(1)}\right)\left(\left(1+\frac{\gamma_{32}P_2}{\alpha+\gamma_{31}P_1^{(2)}}\right)^{\frac{1-\lambda}{\lambda}}-1\right)}.
\end{equation}
%By substituting \eqref{eq_half_noise} into \eqref{eq_half_cf_ts}, we can rewrite the achievable rate as:
%\begin{equation}
%R_{\textrm{ts\_h}}(P) = \frac{1}{2}R_{\textrm{ts}}(P) + %\frac{1}{2}\Gamma\left(\gamma_{31}\beta\frac{P}{\alpha}\right),
%\end{equation}
%\begin{equation}
%R'_{\textrm{ts\_h}}(P) = \frac{1}{2}R'_{\textrm{ts}}(P) + \frac{1}{4}\frac{\gamma_{31}\beta}{1+\frac{\gamma_{31}\beta P}{\alpha}},
%\end{equation}
%where $R_{\textrm{ts}}(P)$ is the achievable rate of the time-sharing CF scheme under the full-duplex mode.
Subsequently, we can have an upper bound by using the definition of the minimum energy per bit. Then as in the full-duplex mode, the above upper bound can be tightened to achieve that in \eqref{eq_half_ebn0_upper} by letting $\alpha$ decrease to $0$ at the same rate as $P$, i.e., $P/\alpha = A$. Moreover, the upper bound in \eqref{eq_half_ebn0_upper} converges to $2\log_e 2/(\gamma_{31}+\gamma_{21})$ when $\gamma_{32}\to\infty$, which implies that the upper bound is asymptotically optimal.
\end{proof}

The comparisons between the time-sharing CF upper bound and the lower bound under both the full-duplex and half-duplex modes are shown in Section \ref{sec_4}. Note that the minimum energy per bit can be clearly improved by time-sharing; hence we do not need to check the second performance measure: the relative spectral efficiency~\cite{verdu}, which is related to the second-order derivative of the achievable rate.

\section{Numerical Results} \label{sec_4}

In this section, we illustrate the performance improvement due to time-sharing in CF with several numerical results. The achievable rate and the minimum energy bit are evaluated under both the full-duplex and half-duplex modes.

The following channel model is considered for all the numerical comparisons. The source, relay, and destination are aligned over a line. The distance between the source and the relay is $d$ ($0<d<1$), and the distance between the source and the destination is $1$. The channel amplitude is inversely proportional to the distance, which means:
\begin{equation}
h_{21} = \frac{1}{d},\ h_{32} = \frac{1}{1-d},\ h_{31} = 1,\quad d
\in (0,1).
\end{equation}

The improvement of achievable rate for various $d$ values is shown in Fig. \ref{fig_full_rate} and Fig. \ref{fig_half_rate} under the full-duplex and half-duplex modes when the active fraction $\alpha$ is optimized for the corresponding time-sharing CF strategies. Here we demonstrate the improvement by drawing the relative rate improvement $(R_{\textrm{ts}}-R_{\textrm{cf}})/R_{\textrm{cf}}$ since the absolute values of the achievable rates are small in the low-SNR regime. As shown in Fig. \ref{fig_full_rate}, the time-sharing CF scheme improves the achievable rate up to $80\%$ under the full-duplex mode when $d = 0.25$ or $0.75$, and up to $75\%$ when $d = 0.5$. Specifically, the relative improvement at $d=0.25$ and $0.75$ are the same since the CF achievable rate under the full-duplex mode is a symmetric function with respect to $d=0.5$. If the relay works under the half-duplex mode, in Fig. \ref{fig_half_rate} we show that the time-sharing CF scheme improves the achievable rate up to $70\%$ when $d=0.25$, up to $55\%$ when $d=0.5$, and up to $49\%$ when $d=0.75$. Unlike the full-duplex mode, the relative improvement decreases if $d$ is larger than $0.5$.

The bounds on the minimum energy per bit are given in Fig. \ref{fig_full_ebn0} and Fig. \ref{fig_half_ebn0} with respect to the source-relay distance $d$. Under the full-duplex mode as shown in Fig. \ref{fig_full_ebn0}, the upper bound can be obviously tightened by the time-sharing CF scheme when the relay is closer to the destination than the source ($d\geq 0.6$). Under the half-duplex mode as shown in Fig. \ref{fig_half_ebn0}, the time-sharing CF scheme tightens the upper bound when the relay is further away from the source ($d\geq 0.8$). We can also see that the upper and lower bounds meet under both the full-duplex and half-duplex modes when $d$ goes to $1$, i.e., $\gamma_{32}\to \infty$.

\section{Conclusion} \label{sec_5}

We investigated the performance of the time-sharing CF scheme in the low-SNR regime over Gaussian relay channels. The achievable rates and the minimum energy per bit were studied under both the full-duplex and half-duplex TDD modes, where we derived the upper and lower bounds on the minimum energy per bit. Furthermore, we showed that the time-sharing CF scheme tightens the upper bound on the minimum energy per bit by letting the active fraction go to zero at the same speed as the transmit power. We also provided some numerical results, which validates the performance improvement due to time-sharing.


\appendices
\section{Proof For Theorem 2}\label{appendix_1}
%\begin{proof}
Based on the max-flow min-cut bound \eqref{eq_half_upperbound} on the system capacity, the minimum energy per bit can be lower-bounded by:
\begin{equation}
\left(\frac{E_b}{N_0}\right)_{\min}^{-} = \min_{0\leq\rho_x\leq 1, 0<\beta\leq 1, 0<\lambda\leq 1} \max \left\{\frac{\log_e 2}{C_{\textrm{1-half}}^{+'}},\frac{\log_e 2}{C_{\textrm{2-half}}^{+'}}\right\},\label{eq_half_lowerbound}
\end{equation}
where the first-order derivatives of $C_{\textrm{1-half}}^{+}$ and $C_{\textrm{2-half}}^{+}$ at $P=0$ can be respectively computed as:
\begin{eqnarray}
C_{\textrm{1-half}}^{+'} &=& \frac{1}{2}\left(\gamma_{31}\beta+\gamma_{32}\left(1-\lambda\right)(1-\beta)+2\left(1-\lambda\right)\rho_x\sqrt{\gamma_{31}\gamma_{32}\beta(1-\beta)}\right),\nonumber\\
C_{\textrm{2-half}}^{+'} &=& \frac{1}{2} \left(\gamma_{31}\beta+\gamma_{21}\lambda\beta-\gamma_{31}\left(1-\lambda\right)\beta\rho_x^2\right).\label{eq_half_lower_deri}
\end{eqnarray}
We can equivalently formulate the optimization problem as:
\begin{equation}
\max_{0\leq\rho_x\leq 1, 0<\beta\leq 1, 0<\lambda\leq 1} \min \left\{C_{\textrm{1-half}}^{+'},C_{\textrm{2-half}}^{+'}\right\}.
\end{equation}

If the relay-destination link gain $\gamma_{32}$ goes to infinity, $C_{\textrm{2-half}}^{+'}$ is always smaller than $C_{\textrm{1-half}}^{+'}$ for $\beta \in (0,1)$ and $\lambda \in (0,1)$, and the lower bound on the energy per bit can be calculated as:
\begin{eqnarray}
\nonumber\left(\frac{E_b}{N_0}\right)_{\min}^{-} &=& \min_{0\leq\rho_x\leq 1, 0<\beta< 1, 0<\lambda< 1} \frac{\log_e 2}{C_{\textrm{2-half}}^{+'}}\\
&=& \frac{2\log_e 2}{\gamma_{31}+\gamma_{21}},
\end{eqnarray}
%where $\epsilon$ is an arbitrarily small positive number.%
which proves the second part of Theorem 2.

When the relay-destination link gain $\gamma_{32}$ is finite, first we calculate the corresponding optimal $\beta$ and $\rho_x$ for an arbitrarily given TDD parameter $\lambda$. %For a given $\beta$, we determine the optimal $\rho_x$ as follows.
If $\beta\leq \gamma_{32}(1-\lambda)/\left(\gamma_{32}(1-\lambda)+\gamma_{21}\lambda\right) \triangleq \beta_1$, the optimal $\rho_x$ is $0$ since $C_{\textrm{1-half}}^{+'}\geq C_{\textrm{2-half}}^{+'}$ and $\min \left\{C_{\textrm{1-half}}^{+'},C_{\textrm{2-half}}^{+'}\right\} = C_{\textrm{2-half}}^{+'}$, where $C_{\textrm{2-half}}^{+'}$ is a monotonically decreasing function over $\rho_x$. In this case, $C_{\textrm{2-half}}^{+'}$ is maximized by $\beta = \beta_1$.

If $\beta \geq \gamma_{32}(1-\lambda)/\left(\gamma_{32}(1-\lambda) + \left(\sqrt{\gamma_{21}\lambda}-\sqrt{\gamma_{31}(1-\lambda)}\right)^2\right) \triangleq \beta_2$, we have $\min \left\{C_{\textrm{1-half}}^{+'},C_{\textrm{2-half}}^{+'}\right\} = C_{\textrm{1-half}}^{+'}$ since $C_{\textrm{2-half}}^{+'}$ is greater than $C_{\textrm{1-half}}^{+'}$ for all $\rho_x \in [0,1]$ (This case only happens when $\lambda > \gamma_{31}/\left(\gamma_{31}+\gamma_{21}\right)$ to ensure that $C_{\textrm{2-half}}^{+'}$ is positive.). Now the optimal $\rho_x$ is $1$ since $C_{\textrm{1-half}}^{+'}$ is a monotonically increasing function over $\rho_x$. The optimal $\beta$ is given as follows by maximizing $C_{\textrm{1-half}}^{+'}\big\vert_{\rho_x = 1}$, which stands for the value of $C_{\textrm{1-half}}^{+'}$ given $\rho_x=1$ where the resulting function is concave over $\beta$:
\begin{equation}
\beta^*=\max \left(\beta_2, \beta_2^* \triangleq \frac{1}{2}+\frac{1}{2}\frac{\gamma_{31}-\gamma_{32}(1-\lambda)}{\sqrt{\left(\gamma_{31}-\gamma_{32}(1-\lambda)\right)^2+4\gamma_{31}\gamma_{32}(1-\lambda)^2}}\right) .
\end{equation}

If $\beta_1 \leq \beta \leq \beta_2$, we let $C_{\textrm{1-half}}^{+'}(\rho_x) = C_{\textrm{2-half}}^{+'}(\rho_x)$ in order to solve the min-max problem since there must exist an unique crossing point. The optimal $\rho_x$ can be obtained as:
\begin{equation}
\rho_x^* = \frac{\sqrt{\gamma_{21}\lambda\beta/(1-\lambda)}-\sqrt{\gamma_{32}(1-\beta)}}{\sqrt{\gamma_{31}\beta}}.\label{eq_half_rho_x}
\end{equation}
By substituting \eqref{eq_half_rho_x} into \eqref{eq_half_lower_deri}, we can maximize $C_{\textrm{2-half}}^{+'}$ for $\beta \in [\beta_1,\beta_2]$, which leads to the following optimal $\beta$:
\begin{equation}
\beta^*=\label{eq_half_beta}\min \left(\beta_2,\beta_3^* \triangleq \frac{1}{2}+\frac{1}{2}\frac{\gamma_{31}+\gamma_{32}(1-\lambda)}{\sqrt{(\gamma_{31}+\gamma_{32}(1-\lambda))^2+4\gamma_{21}\gamma_{32}\lambda(1-\lambda)}}\right).
\end{equation}
When $\beta_3^*$ is optimal for a given $\lambda$, the lower bound on the minimum energy per bit can be derived by substituting $\beta_3^*$ and the corresponding $\rho_x^*$ into \eqref{eq_half_lowerbound} and \eqref{eq_half_lower_deri}:
\begin{eqnarray}\label{eq_half_lowerbound_3}
\nonumber\left(\frac{E_b}{N_0}\right)_{\min}^{-} &=& \min_{0<\lambda\leq 1} \frac{\log_e 2}{C_{\textrm{2-half}}^{+'}}\\
&=& \min_{0<\lambda\leq 1} \frac{4\log_e 2}{\gamma_{31}-\gamma_{32}(1-\lambda) + \sqrt{\left(\gamma_{31}+\gamma_{32}(1-\lambda)\right)^2+4\gamma_{21}\gamma_{32}\lambda(1-\lambda)}},
\end{eqnarray}
where the denominator is a concave function over $\lambda \in (0,1]$. Therefore, the optimal $\lambda$ can be obtained by maximizing the denominator in \eqref{eq_half_lowerbound_3}:
\begin{equation}\label{eq_opt_lambda}
\lambda_{\textrm{opt}} = \frac{\gamma_{31}+\gamma_{32}-2\gamma_{21}-\sqrt{\gamma_{31}^2+\gamma_{31}\gamma_{32}+\gamma_{21}\gamma_{32}}}{\gamma_{32}-4\gamma_{21}}
\end{equation}
Correspondingly, we can calculate the lower bound on the minimum energy per bit by substituting $\lambda_{\textrm{opt}}$ into \eqref{eq_half_lowerbound_3}:
\begin{equation}\label{eq_half_lowerbound_4}
\left(\frac{E_b}{N_0}\right)_{\min}^{-} = \frac{2(\gamma_{32}-4\gamma_{21})\log_e 2 }{\gamma_{31}\gamma_{32}-2\gamma_{21}\gamma_{31}+\gamma_{21}\gamma_{32}-2\gamma_{21}\sqrt{\gamma_{31}^2+\gamma_{31}\gamma_{32}+\gamma_{21}\gamma_{32}}}.
\end{equation}

Now we summarize the solution by combining all the cases that we discussed above.
\begin{enumerate}
\item For $\lambda \in (0,\gamma_{31}/(\gamma_{31}+\gamma_{21})]$, $C_{\textrm{2-half}}^{+'}$ is not ensured to be larger than $C_{\textrm{1-half}}^{+'}$ for all $\rho_x \in [0,1]$. Accordingly, the lower bound is $\left(\frac{E_b}{N_0}\right)_{\min}^{-}$ in \eqref{eq_half_lowerbound_4} if $\lambda_{\textrm{opt}}\in (0, \gamma_{31}/(\gamma_{31}+\gamma_{21})]$. Otherwise, the lower bound is $\frac{2\log_e 2}{C_{\textrm{1-half}}^{+'}}\big |_{\lambda=\gamma_{31}/(\gamma_{31}+\gamma_{21}),\beta_{3}^*,\rho_x^*}$.
\item For $\lambda \in (\gamma_{31}/(\gamma_{31}+\gamma_{21}), 1]$, one-dimensional searching can be utilized to obtain the lower bound. Specifically, for a given $\lambda \in (\gamma_{31}/(\gamma_{31}+\gamma_{21}), 1]$, we obtain an intermediate lower bound by taking the minimum among the ones derived in the following four cases. Then we search over all possible $\lambda \in (\gamma_{31}/(\gamma_{31}+\gamma_{21}), 1]$ to have the lower bound for case 2).
    \begin{itemize}
    \item If $\beta_2\leq\beta_2^*,\ \beta_3^*$, the corresponding lower bound is $\frac{2\log_e 2}{C_{\textrm{1-half}}^{+'}}\big |_{\beta_2^*,\rho_x=1}$;
    \item If $\beta_3^*<\beta_2\leq\beta_2^*$, the corresponding lower bound is $\frac{2\log_e 2}{\max\left(C_{\textrm{1-half}}^{+'}|_{\beta_2^*,\rho_x=1}, C_{\textrm{1-half}}^{+'}|_{\beta_3^*,\rho_x^*}\right)}$;
    \item If $\beta_2^*<\beta_2\leq\beta_3^*$, the corresponding lower bound is $\frac{2\log_e 2}{C_{\textrm{1-half}}^{+'}}\big |_{\beta_2,\rho_x=1}$;
    \item Otherwise, $\beta_2^*,\ \beta_3^*<\beta_2$, the corresponding lower bound is $\frac{2\log_e 2}{C_{\textrm{1-half}}^{+'}}\big |_{\beta_3^*,\rho_x^*}$.
    \end{itemize}
\item The global lower bound can be calculated as the minimum value between that for case 1) of $\lambda \in (0,\gamma_{31}/(\gamma_{31}+\gamma_{21})]$ and that for case 2) of $\lambda \in (\gamma_{31}/(\gamma_{31}+\gamma_{21}), 1]$.
\end{enumerate}
%\end{proof}

%\section*{Acknowledgment}
% optional entry into table of contents (if used)
%\addcontentsline{toc}{section}{Acknowledgment}
%The authors would like to thank various sponsors for supporting their research.
% In particular, we thank the TPC chairs of ISIT 2006 for
% providing the \LaTeX\ templates for paper submission.



% http://www.ctan.org/tex-archive/macros/latex/contrib/supported/IEEEtran/bibtex
%\bibliographystyle{IEEEtran.bst}

\newpage
\begin{figure}[!t]
\centering
% for double column submission
`\includegraphics[width=3.5in]{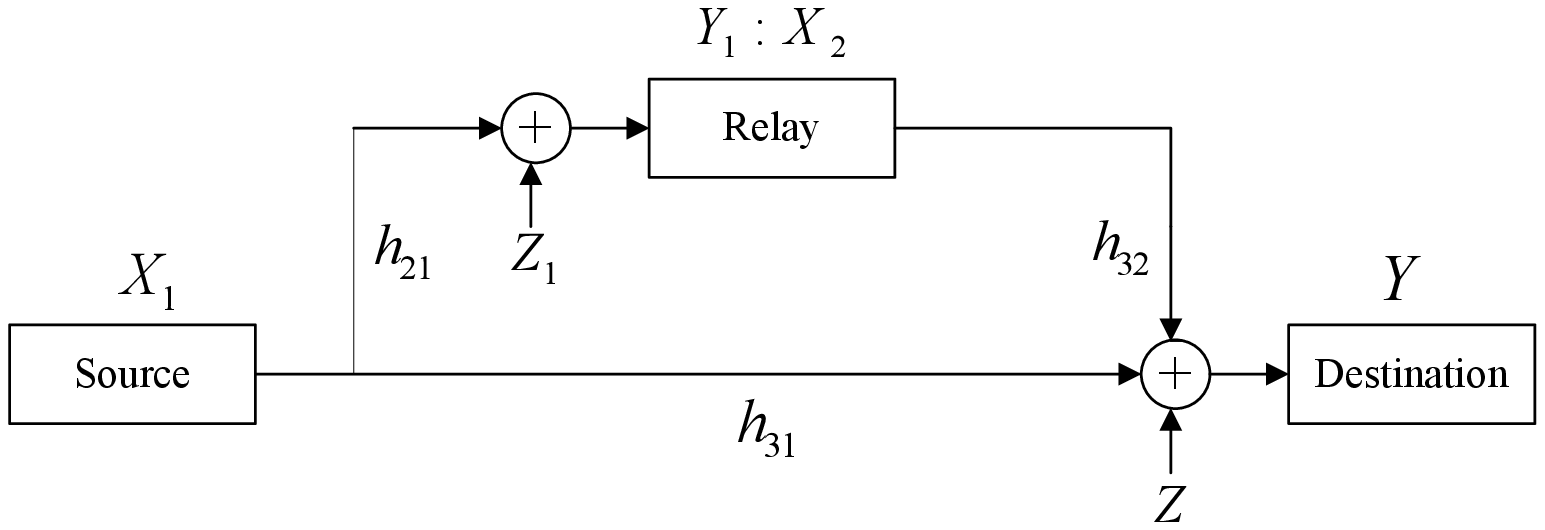}
\caption{The Relay Channel Model.} \label{fig_relay_channel_model}
\end{figure}

\begin{figure}[!t]
\centering
% for double column submission
\includegraphics[width=3.5in]{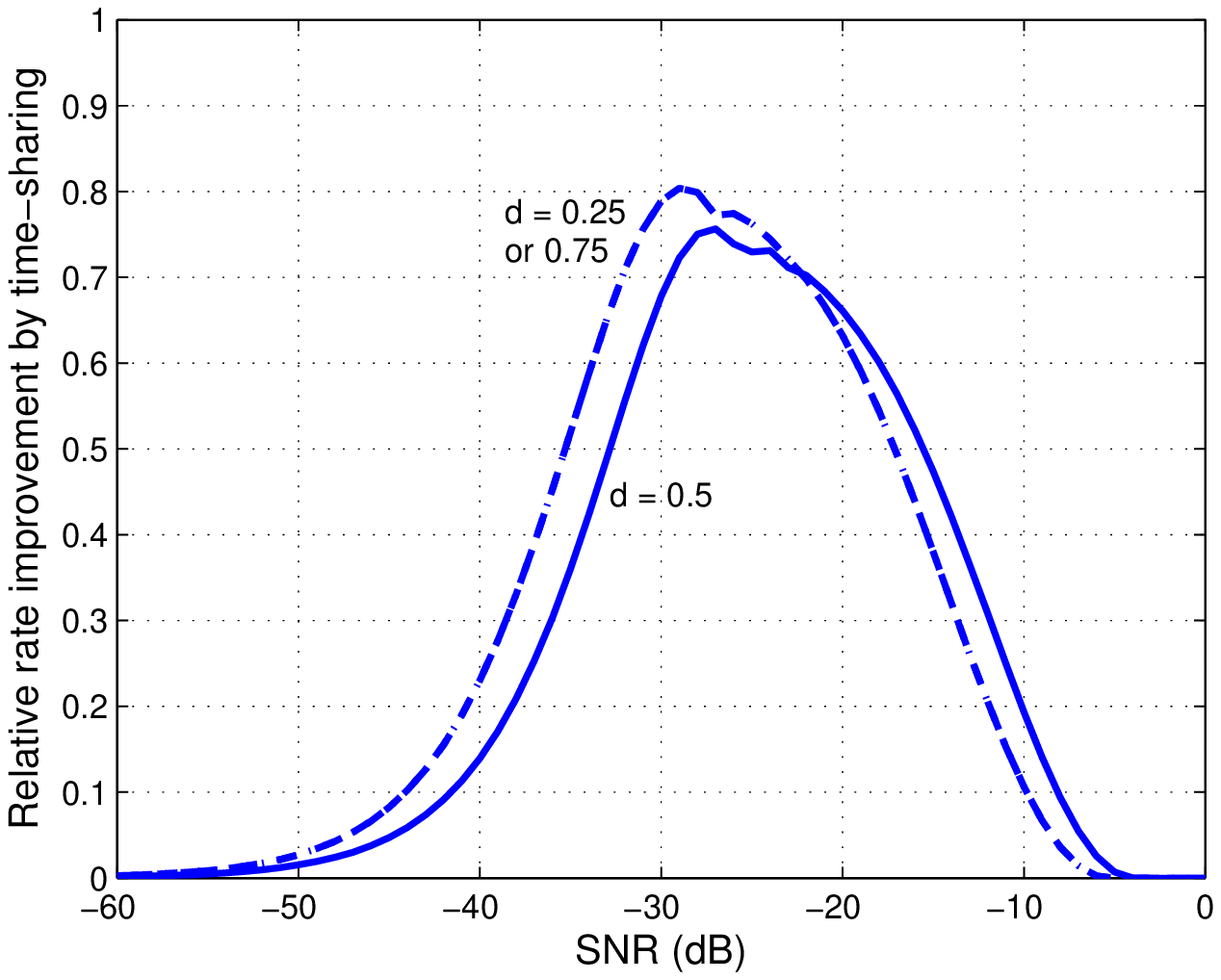}
\caption{The relative rate improvement of time-sharing CF under full-duplex mode.} \label{fig_full_rate}
\end{figure}

\begin{figure}[!t]
\centering
% for double column submission
\includegraphics[width=3.5in]{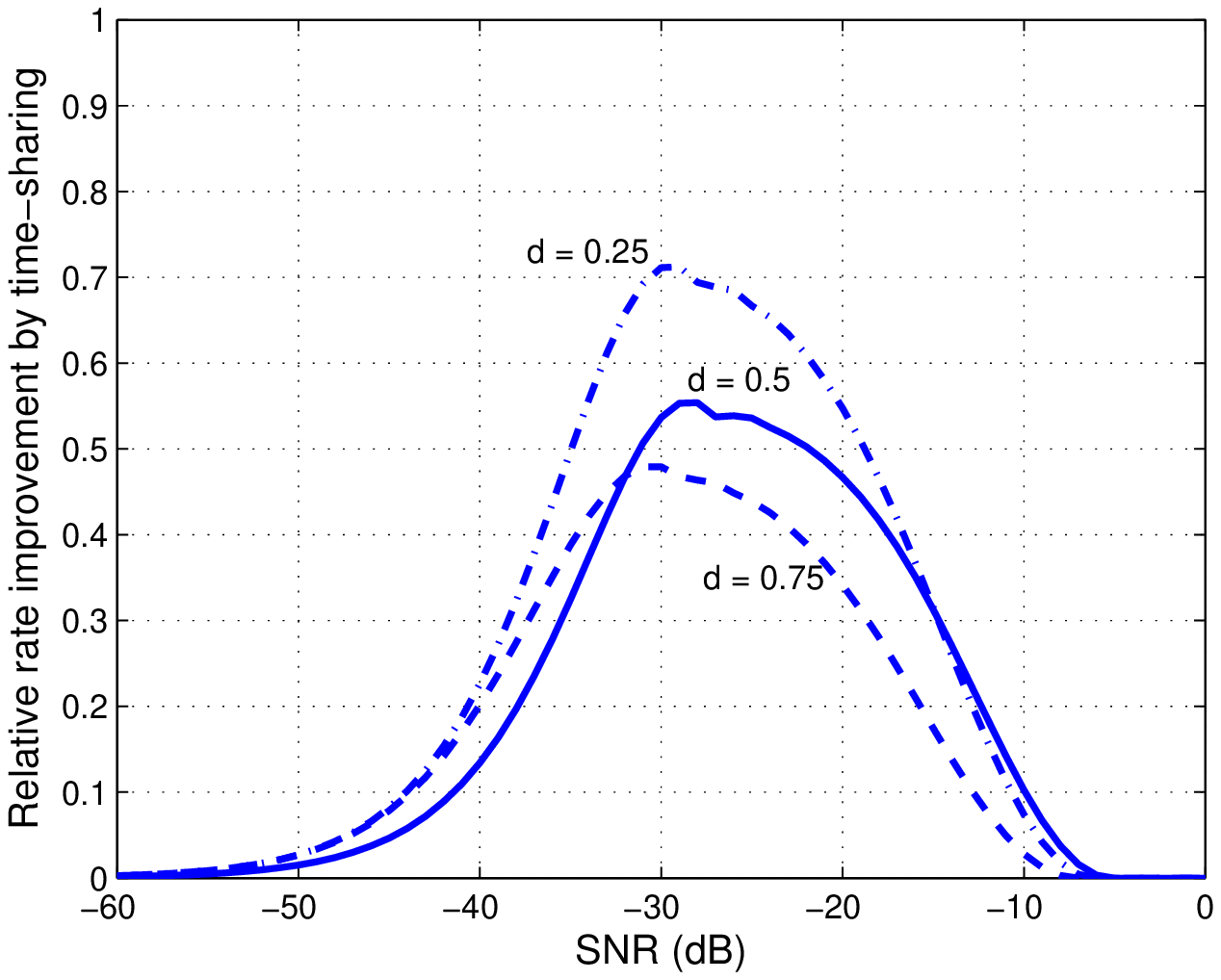}
\caption{The relative rate improvement of time-sharing CF under half-duplex mode.} \label{fig_half_rate}
\end{figure}

\begin{figure}[!t]
\centering
% for double column submission
\includegraphics[width=3.5in]{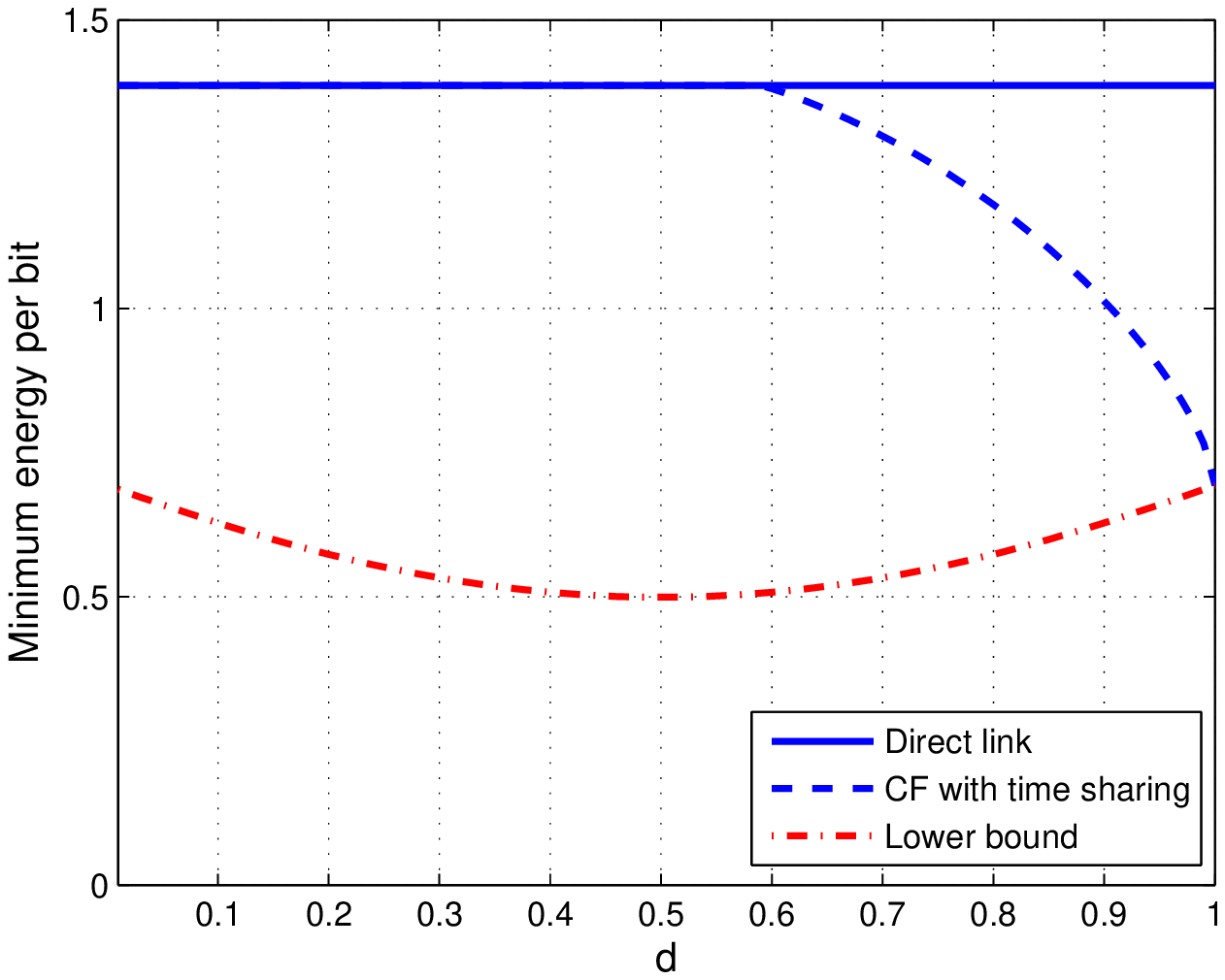}
\caption{The bounds on the minimum energy per bit under full-duplex mode}\label{fig_full_ebn0}
\end{figure}

\begin{figure}[!t]
\centering
% for double column submission
\includegraphics[width=3.5in]{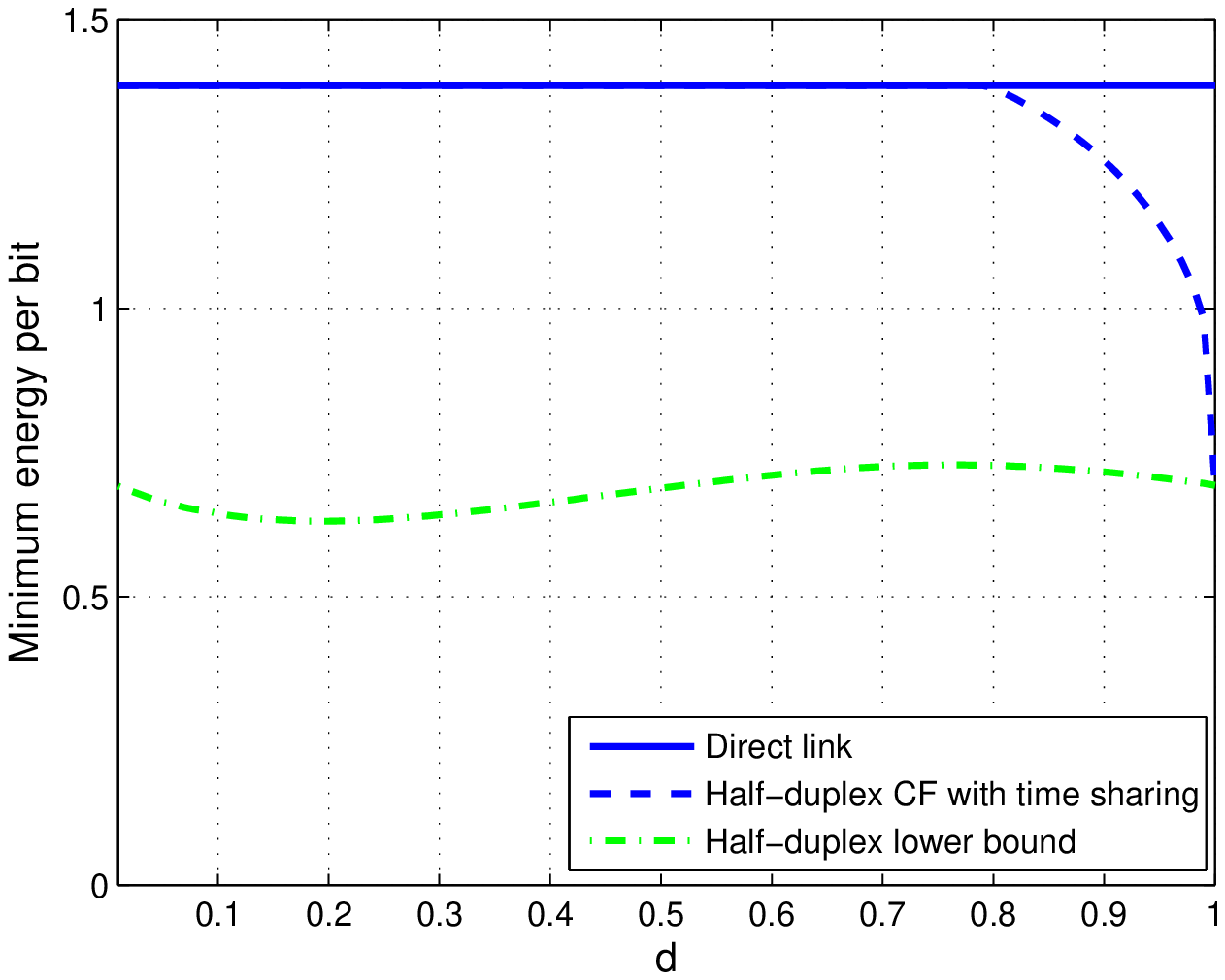}
\caption{The bounds on the minimum energy per bit under half-duplex mode}\label{fig_half_ebn0}
\end{figure}

\end{document}